\documentclass[aps,physrev,notitlepage,11pt]{revtex4-2}
\pdfoutput=1
\usepackage[utf8]{inputenc}
\usepackage[english]{babel}
\usepackage[T1]{fontenc}
\usepackage{amsthm}
\usepackage{amsmath}
\usepackage{latexsym}
\usepackage{amsfonts}
\usepackage{amssymb}
\usepackage{color}
\usepackage{bbm,dsfont}
\usepackage{graphicx}
\usepackage{subfigure}
\usepackage{mathrsfs}
\usepackage{mathbbol}
\usepackage{enumerate}
\usepackage{MnSymbol}
\usepackage{epstopdf}
\usepackage{comment}
\usepackage{hyperref}
\hypersetup{citecolor=blue}
\hypersetup{colorlinks=true}
\hypersetup{linkcolor=blue}
\hypersetup{urlcolor=blue}
\usepackage{thmtools, thm-restate}
\newtheorem{proposition}{Proposition}
\declaretheorem{theorem}
\newtheorem{lemma}{Lemma}
\newtheorem*{lemma*}{Lemma}
\newtheorem{corollary}{Corollary}

\theoremstyle{definition}

\newtheorem{example}{Example}
\newtheorem{definition}{Definition}

\usepackage[dvipsnames]{xcolor}

\newcommand{\real}{\mathbb R} 
\newcommand{\half}{\tfrac{1}{2}} 

\newcommand{\lh}{\mathcal{L(H)}} 
\newcommand{\lhp}{\lh_+} 
\newcommand{\ket}[1]{|#1\rangle} 
\newcommand{\bra}[1]{\langle#1|} 
\newcommand{\no}[1]{\left\|#1\right\|} 
\newcommand{\tr}[1]{\mathrm{tr}\left[#1\right]} 
\newcommand{\ptr}[2]{\mathrm{tr}_{#1}[#2]} 
\newcommand{\id}{\mathbbm{1}} 

\newcommand{\A}{\mathsf{A}}
\newcommand{\B}{\mathsf{B}}
\newcommand{\C}{\mathsf{C}}
\newcommand{\G}{\mathsf{G}}
\newcommand{\M}{\mathsf{M}}
\newcommand{\X}{\mathsf{X}}
\newcommand{\Y}{\mathsf{Y}}
\newcommand{\Z}{\mathsf{Z}}

\newcommand{\coin}{\mathcal{C}} 

\newcommand{\state}{\mathcal{S}} 
\newcommand{\obs}{\mathcal{O}} 

\newcommand{\lin}[1]{{\rm span}\left( #1 \right)}

\newcommand{\tests}{\mathcal{T}} 
\newcommand{\testa}{\mathcal{A}} 
\newcommand{\testb}{\mathcal{B}} 

\newcommand{\fix}{\mathrm{Fix}}

\def\Ha{\mathcal{H}}
\def\Ka{\mathcal{K}}

\def\1{\mathbb{1}}

\DeclareMathOperator{\swap}{SWAP}

\newcommand{\ketbra}[1]{\ket{#1} \! \bra{#1}}

\begin{document}

\title{Dispensing of quantum information beyond no-broadcasting theorem - is it possible to broadcast anything genuinely quantum?}

\begin{abstract}
No-broadcasting theorem is one of the most fundamental results in quantum information theory; it guarantees that the simplest attacks on any quantum protocol, based on eavesdropping and copying of quantum information, are impossible. Due to the fundamental importance of the no-broadcasting theorem, it is essential to understand the exact boundaries of this limitation. We generalize the standard definition of broadcasting by restricting the set of states which we want to broadcast and restricting the sets of measurements which we use to test the broadcasting. We show that in some of the investigated cases broadcasting is equivalent to commutativity, while in other cases commutativity is not necessary.
\end{abstract}

\author{Teiko Heinosaari}
\email{teiko.heinosaari@utu.fi}
\affiliation{Quantum Algorithms and Software, VTT Technical Research Centre of Finland Ltd, Finland}
\affiliation{Department of Physics and Astronomy, University of Turku, Turku 20014, Finland}

\author{Anna Jen\v{c}ov\'{a}}
\email{jenca@mat.savba.sk}
\affiliation{Mathematical Institute, Slovak Academy of Sciences, Bratislava, Slovakia}

\author{Martin Pl\'{a}vala}
\email{martin.plavala@uni-siegen.de}
\affiliation{Naturwissenschaftlich-Technische Fakult\"{a}t, Universit\"{a}t Siegen, 57068 Siegen, Germany}

\maketitle

\section{Introduction}

The no-cloning theorem is perhaps the most famous limitation in quantum information processing. The impossibility to duplicate an unknown quantum state makes a fundamental difference between classical and quantum information processing. The no-cloning theorem is nowadays seen not only as an obstacle but a fact that enables e.g. secure communication protocols. Namely, since an unknown quantum state cannot be perfectly duplicated it follows that a potential eavesdropper on a quantum communication channel cannot capture messages without being detected. The no-cloning theorem is not an isolated feature of quantum theory but links to other impossibility statements such as the no-information-without-disturbance theorem and the non-existence of joint measurement for arbitrary pairs of measurements \cite{QI01Werner}. In that perspective, the violation of Bell inequalities can be seen as an experimental proof of the no-cloning theorem.

Later, an important distinction between cloning and broadcasting was made \cite{Barnumetal96} and in the present work we concentrate on the latter concept, hence we recall their difference. A (hypothetical) perfect cloning device takes an unknown pure state $\rho$ as an input and gives a composite system in a joint state $\rho\otimes\rho$ as an output. Even in a classical theory this condition cannot hold for all mixed states, therefore it makes more sense to pose it only for pure quantum states. In fact, the originally presented version of the no-cloning theorem formulates cloning for pure states and derives a contradiction from their superposition structure \cite{WoZu82}. The defining condition for a perfect broadcasting device is weaker than cloning; it is only required that the output is a joint state $\omega$ of a bipartite system that has the duplicate of the initial state $\rho$ as its both marginal states, i.e., the partial traces of $\omega$ are $\ptr{1}{\omega}=\ptr{2}{\omega}=\rho$. The broadcasting condition captures the essence of the possibility to duplicate classical information as in a classical theory; the perfect cloning device for pure states is also a perfect broadcasting device that broadcasts all classical states, pure and mixed. In quantum theory the impossibility of universal and perfect broadcasting follows from the no-cloning theorem. This is due to the fact that a joint state with pure marginal states is necessarily a product state, hence for pure states broadcasting is the same as cloning. Therefore, taken only as categorical no-go statements, the no-cloning theorem and the no-broadcasting theorem are equivalent. We emphasize that the previous statements and also all the later developments in this work assume that there is one input copy. If there are more input copies, then no-cloning and no-broadcasting become separated and broadcasting of qubit states is, in fact, possible if there are four or more input copies available \cite{d2005superbroadcasting,buscemi2006universal}.

The separation between cloning and broadcasting becomes relevant when one considers approximate, non-universal or otherwise imperfect scenarios. The underlying theme is to find and characterize possible quantum devices that function as a cloning or broadcasting device in some approximate or limited manner. The sole no-go statement does not prevent the existence of a device with arbitrarily small nonzero deviation from the hypothetical cloning or broadcasting device. Since the existence of such an almost perfect device would evidently ruin the essence and practical consequences of the no-go theorems, this topic is important and various approximate scenarios have been exhaustively studied earlier \cite{ScIbGiAc05,CeFi06}. In approximate scenarios the distinction between cloning and broadcasting is not sharp as an actual device can be compared to both of these hypothetical devices and it simulates them imperfectly. Roughly speaking, if the aim is optimal approximate broadcasting, then the quality of output states is tested on the individual subsystems and no attention is paid to their joint state. In contrast to approximate scenarios, one can consider perfectly accurate but non-universal cloning or broadcasting devices, meaning that the input state is not completely arbitrary but belongs to some known subset of states. In this case the distinction between cloning and broadcasting is clear since the performance of the device is required to be perfectly that of the hypothetical device, just not on all states. It has been shown that a set of quantum states that can be broadcasted with a single quantum channel is contained in the simplex generated by a set of distinguishable states. In other words, maximal subsets of states that can be broadcasted are replicas of classical state spaces in the given quantum state space.

In the current work we generalize the setting of accurate but limited broadcasting. We formulate the concept of a broadcasting test, where the capability of a quantum channel to broadcast is tested with a setup consisting of some specified test states and test measurements. This formulation covers several interesting scenarios as special cases, some of which have been investigated earlier and some new arising naturally from the introduced framework. We prove that the defining condition of a broadcasting test reduces to a mathematical condition on factorization maps, determined by the corresponding test measurements. The derived condition on factorization maps avoids irrelevant details and provides tools to study any given broadcasting test. Further, the formulation of a broadcasting test in terms of factorization maps leads to a new viewpoint, where broadcasting can be seen as a certain kind of congruency relation on quantum channels. We prove that this relation is strictly stronger than the compatibility relation on quantum channels. In the final section we will study broadcasting in two specific cases: when the sets of measurements used to test the broadcasting are restricted and when the set of states and set of measurements on one side are restricted. We show that in all cases commutativity is sufficient for broadcasting, but, surprisingly, we also show that commutativity is strictly not necessary by constructing concrete examples of scenarios which are broadcastable, but not commutative. This complements the earlier result on broadcastable set of states where commutativity was often shown to be necessary. The general concept of a broadcasting test allows to study intermediate cases, where limitations are both of states and measurements.

\section{Broadcasting tests}\label{sec:tests}

\begin{figure}
\centering
\subfigure[]
{
\includegraphics[height=2.8cm]{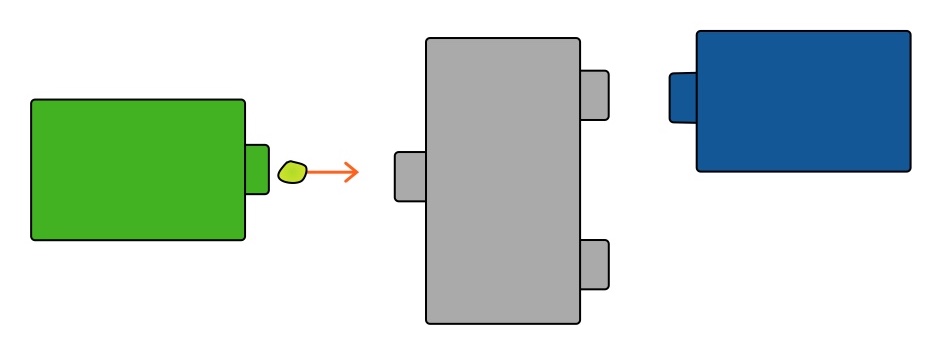}
}
\hspace{0.1cm}
\subfigure[]
{
\includegraphics[height=2.8cm]{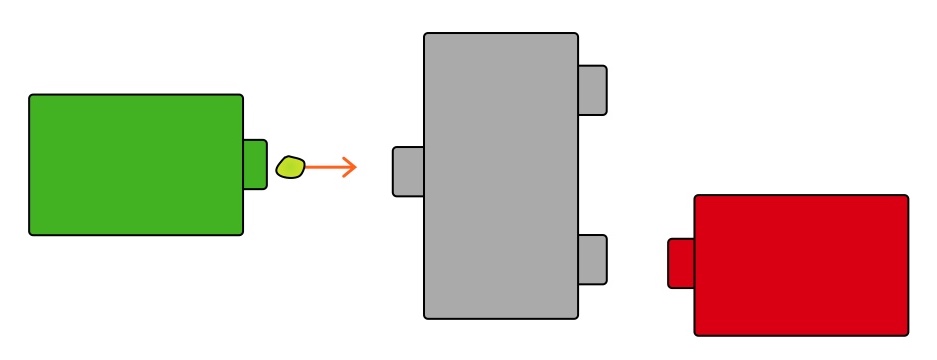}
}
\hspace{2cm}
\subfigure[]
{
\includegraphics[height=3.5cm]{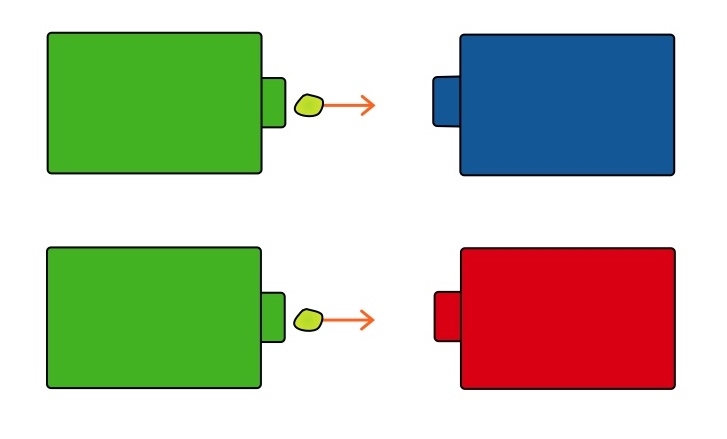}
}
\caption{\label{fig:test}(a),(b) In a broadcasting test some chosen test states are produced with a preparator (green device), inserted to a broadcasting channel (gray device) and measurements (blue and red devices) are then performed on the two output systems. (c) The measurement outcome distributions are compared to the situation where the same measurements are performed independently. If these two situations cannot be differentiated, then the channel passes the broadcasting test.}
\end{figure}

By a quantum channel we mean a completely positive and trace preserving linear map from a quantum state space to another quantum state space. These are the physically realizable input-output processes. Let us consider a scenario where a quantum channel $\Lambda:\state \to \state \otimes \state$ is planned to be used to broadcast an unknown state to two parties. Here $\state$ is a fixed state space of $d$-dimensional quantum system. Generally, a channel from $\state$ to $\state \otimes \state$ (or more generally to $\state^{\otimes n}$) is called a broadcasting channel, irrespective of its action. Here we use $\state \otimes \state$ to denote the tensor product of quantum state spaces, $\state \otimes \state$ is the state space consisting of density matrices on the tensor product of the underlying Hilbert spaces. The broadcasting condition is
\begin{equation}\label{eq:broad}
\ptr{1}{\Lambda(\rho)}=\ptr{2}{\Lambda(\rho)}=\rho
\end{equation}
and due to the no-broadcasting theorem $\Lambda$ cannot satisfy this for all states $\rho$. However, it may work like that in some limited setting. To test the functioning of $\Lambda$, we insert some test states as inputs and perform some measurements on the output states. The test states and measurements may be completely up to our choice, or they may be determined by some constrains e.g. of a communication scenario. In any case, we assume that the test states and measurements, although arbitrary, are fully known. A measurement is mathematically described as a positive operator valued measure (POVM), i.e., for each possible measurement outcome $i$ is assigned a positive operator $\A_i$ and $\sum_i \A_i = \id$. The measurement outcome distribution in a state $\rho$ is given as $\A_i(\rho):=\tr{\rho \A_i}$ for all $i$ and we will use $\obs$ to denote the set of all measurements on $\state$.

If a state $\rho$ is inserted to a broadcasting channel $\Lambda$ and $\A$ and $\B$ are measurements on the output sides, then $\Lambda$ functions in the intended way if the measurement outcome probabilities are the same for the input state $\rho$ and for the transformed states $\ptr{2}{\Lambda(\rho)}$ and $\ptr{1}{\Lambda(\rho)}$, i.e., $\A(\ptr{2}{\Lambda(\rho)}) =\A(\rho)$ and $\B(\ptr{1}{\Lambda(\rho)}) =\B(\rho)$. This motivates the following definitions, illustrated in Fig.~\ref{fig:test}.

\begin{definition}\label{def:broad}
A \emph{broadcasting test} is a triple $(\tests,\testa,\testb)$, where $\tests\subseteq\state$ is a collection of tests states and $\testa,\testb\subseteq\obs$ are collections of test measurements.
We say that a channel $\Lambda:\state \to \state \otimes \state$ \emph{passes the broadcasting test} $(\tests,\testa,\testb)$ if
\begin{align}
\A(\ptr{2}{\Lambda(\rho)}) =\A(\rho) \label{eq:broad-A} \\
\B(\ptr{1}{\Lambda(\rho)}) =\B(\rho) \label{eq:broad-B}
\end{align}
for all $\rho\in\tests$, $\A\in\testa$ and $\B\in\testb$.
If there exists a channel $\Lambda$ that passes a broadcasting test $(\tests,\testa,\testb)$, then we say that the triple $(\tests,\testa,\testb)$ is \emph{broadcastable}.
\end{definition}

Clearly, universal perfect broadcasting corresponds to the choices $\tests=\state$ and $\testa=\testb=\obs$ and by the no-broadcasting theorem there is no channel that would pass the broadcasting test $(\state,\obs,\obs)$. The question is then to find the conditions when a triple $(\tests,\testa,\testb)$ is broadcastable. It is instructive to separate the following special cases:

\begin{enumerate}[(a)]

\item \emph{Broadcasting of a subset of states}. In this case test measurements are arbitrary but test states are from some limited subset of states, i.e., $\tests\subset\state$, $\testa=\testb=\obs$. The conditions \eqref{eq:broad-A} and \eqref{eq:broad-B} can then be written as \eqref{eq:broad}. It is customary to say that a subset of states $\tests$ is broadcastable when $(\tests,\obs,\obs)$ is broadcastable in the sense of Def.~\ref{def:broad}. It has been proven in \cite{BaBaLeWi07} that this is the case if and only if $\tests$ lies in a simplex generated by jointly distinguishable states. For instance, orthogonal pure states are distinguishable and hence a subset consisting of their convex mixtures is broadcastable.

\item \emph{Broadcasting of a subset of measurements}. This is the analog of the previous case, but now test states are arbitrary while test measurements are limited, i.e., $\tests=\state$, $\testa=\testb\subset\obs$. We also say that a subset $\testa$ is broadcastable when $(\state,\testa,\testa)$ is broadcastable in the sense of Def.~\ref{def:broad}. In the case of $\testa$ consisting of two measurements this reduces to the broadcastability relation studied in \cite{Heinosaari16}, where it was introduced as a strong form of compatibility and broadcastable pairs of qubit measurements were characterized. We present a full characterization of broadcastable sets of measurements in Sec.~\ref{sec:special}, we show that in this case commutativity is sufficient but not necessary for broadcasting.

\item \emph{One-side broadcasting of subsets of measurements}. When thinking about the previous scenario, there is no need for the subsets $\testa$ and $\testb$ to be identical as we can perform different kind of measurements on the two outputs. In the case of pairs of measurements this was called one-side broadcasting in \cite{Heinosaari16} and it was shown that two qubit measurements are one-side broadcastable if and only if they are mutually commuting. We prove in Sec.~\ref{sec:special} that mutual commutativity is sufficient but not necessary in general.

\item \emph{Perfect transmission of a subset of states}. If $\tests\subset\state$ and $\testb=\obs$, then the test states are perfectly transmitted to the $\testb$-side. This scenario is more general than broadcasting of subsets of states since we may have $\testa \subset \obs$, meaning that on the $\testa$-side we make only a partial test. The scenario can correspond e.g. to eavesdropping, where an eavesdropper does not want to leave any traces from the intervenience. We present a characterization of broadcastable test states $\tests$ and test measurement $\testa$ in Sec.~\ref{sec:special}, we show that in this case commutativity is necessary for broadcasting.

\end{enumerate}

\section{Reformulation of broadcasting tests}\label{sec:reformulation}

We begin this section with a simple observation. The equations \eqref{eq:broad-A} and \eqref{eq:broad-B} are linear in $\rho$, $\A$ and $\B$. Therefore, if a channel $\Lambda$ passes a broadcasting test $(\tests,\testa,\testb)$, then it also passes the broadcasting test for all states in $\lin{\tests}$ and all measurements in $\lin{\testa}$ and $\lin{\testb}$. Here $\lin{X}$ denotes the linear span of a set $X$.

In the following we present a mathematically equivalent form of broadcasting tests. To do this, we introduce an equivalence relation that partitions the state space according to a particular test. Firstly, we say that a measurement $\A$ \emph{distinguishes} two states $\rho$ and $\sigma$ if $\A(\rho) \neq \A(\sigma)$, meaning that the measurement outcome distributions for these states are different. We further say that a subset $\testa\subseteq \obs$ distinguishes two states $\rho$ and $\sigma$ if $\A(\rho) \neq \A(\sigma)$ for some $\A\in\testa$. For $\rho, \sigma \in \state$ and $\testa\subseteq\obs$, we denote $\rho \approx_\testa \sigma$ and say that $\rho$ and $\sigma$ are \emph{$\testa$-equivalent} if $\A(\rho) = \A(\sigma)$ for all $\A \in \testa$. Hence, two states are deemed $\testa$-equivalent if no measurement in $\testa$ can distinguish between them.

A collection $\testa$ is called \emph{informationally complete} if it distinguishes all pairs of states. If $\testa$ is informationally complete, then clearly $(\tests,\testa,\testb)$ is broadcastable if and only if $(\tests,\obs,\testb)$ is broadcastable. In other words, if we would use informationally complete subsets of measurements in both output sides, then we have the state broadcasting scenario (case (b) in Sec.~\ref{sec:tests}).

For any $\testa\subseteq\obs$, we denote the equivalence classes in the equivalence relation $\approx_\testa$ as
\begin{equation}
[\rho]_\testa = \{ \sigma \in \state : \rho \approx_\testa \sigma \} \, .
\end{equation}
We further denote $\tests_{\testa} = \{ [\rho]_\testa : \rho \in \tests \}$ and $\state_{\testa} = \{ [\rho]_\testa : \rho \in \state \}$ the sets of all equivalence classes in $\tests$ and $\state$, respectively.
It follows that we can introduce the factorization map $F_\testa : \state \to \state_{\testa}$ as
\begin{equation}
F_\testa(\rho) = [\rho]_\testa \, .
\end{equation}
Further, for any $\A \in \testa$ we define a map $\A'$ on $\state_{\testa}$ as
\begin{equation}
\A' ( [\rho]_\testa ) = \A(\rho) \, .
\end{equation}
This is well-defined since by construction every $\A \in \testa$ is constant on the equivalence classes $[\rho]_\testa$.
We therefore have
\begin{equation}
\A (\rho) = (\A' \circ F_\testa)(\rho) \label{eq:reform-factor-obseravable}
\end{equation}
for all $\rho \in \state$.

\begin{figure}
\centering
\subfigure[]
{
\includegraphics[height=4.5cm]{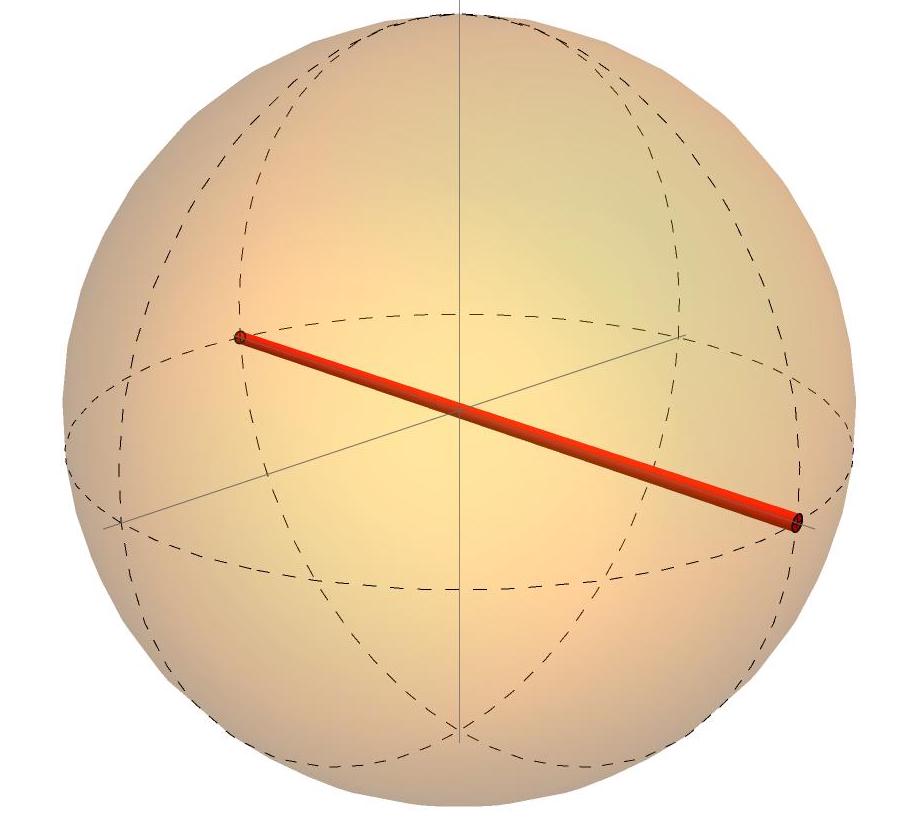}
}
\hspace{1.5cm}
\subfigure[]
{
\includegraphics[height=4.5cm]{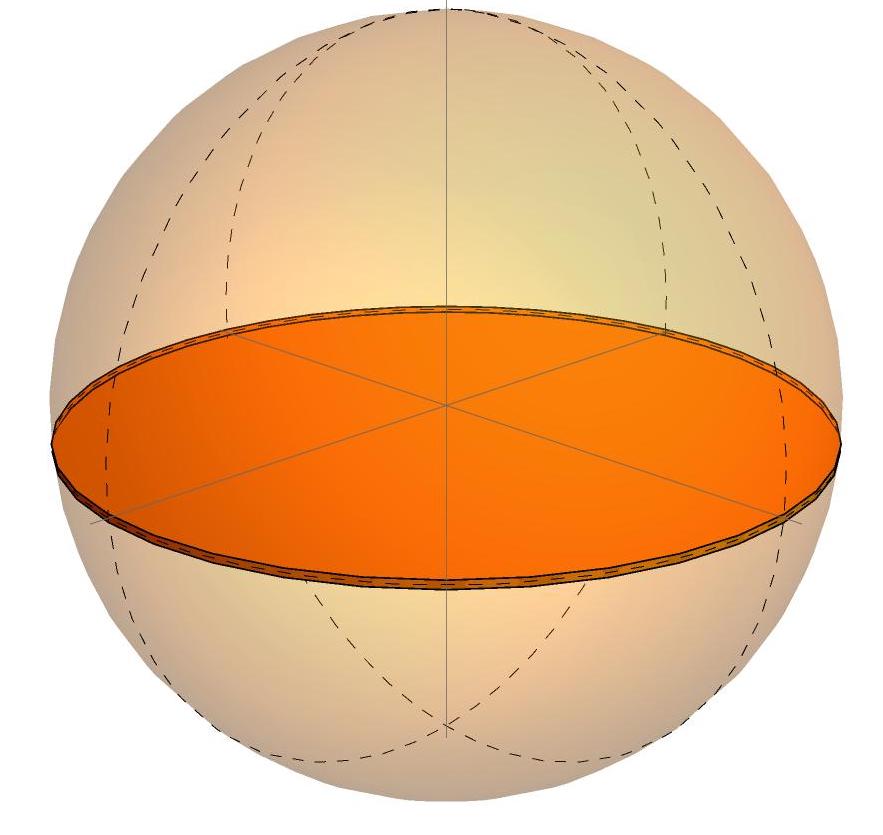}

}
\caption{ \label{fig:bloch} The qubit state space can be seen as a three dimensional ball and then factorization maps are shrinking the ball in some specific way, depending on the subset of measurement. (a) The set $\state_\X$ can be identified with the interval $[-1,1]$ while (b) the set $\state_{\{\X,\Y\}}$ can be identified with the disk $\{ (x,y) : x^2+y^2 \leq 1\}$. }
\end{figure}

\begin{example}\label{ex:factor-qubit}(\emph{Factorization maps on a qubit}.)
Let us consider a qubit system.
The qubit states can be parametrized with Bloch vectors $\vec{r}\in\real^3$, $\no{\vec{r}}\leq 1$, via the correspondence $\rho=\half (\id + r_x\sigma_x + r_y\sigma_y + r_z\sigma_z )$, where $\sigma_x,\sigma_y,\sigma_z$ are the Pauli matrices. We denote by $\X$, $\Y$, $\Z$ the measurements that measure the components of the Bloch vector $\vec{r}$. For instance, $\X(\rho)=(\half(1 + r_x),\half (1-r_x))$. We can identify the set $\state_\X$ of all equivalence classes with the interval $[-1,1]$ and the factorization map $F_\X$ is $F_\X(\rho)=\tr{\rho \sigma_x}$ (see Fig.~\ref{fig:bloch}(a)). In this representation the map $\X'$ is given by $\X'(p)=(\half (1+p),\half(1-p))$. The description is analogous for $\Y$ and $\Z$, but becomes different if $\testa$ consists of two measurements, say $\X$ and $\Y$. We can identify the set $\state_{\{\X,\Y\}}$ with the disk $\{ (x,y) : x^2+y^2 \leq 1\}$ and the factorization map $F_{\{\X,\Y\}}$ is then $F_{\{\X,\Y\}}(\rho) = (\tr{\rho \sigma_x},\tr{\rho \sigma_y})$ (see Fig.~\ref{fig:bloch}(b)). The maps $\X'$ and $\Y'$ are now given by $\X'(p,q)=(\half (1+p),\half(1-p))$ and $\Y'(p,q)=(\half (1+q),\half(1-q))$. If $\testa$ consists of all three measurements $\X$, $\Y$ and $\Z$, then the set is informationally complete. The factorization map $F_{\{\X,\Y,\Z\}}$ is then simply the bijection between the qubit state space and the unit ball.
\end{example}

With the previous concepts we now express the broadcasting test condition in terms of factorization maps as follows.

\begin{proposition} \label{prop:reform-broadcast}
A channel $\Lambda: \state \to \state \otimes \state$ passes a broadcasting test $(\tests, \testa, \testb)$ if and only if
\begin{align}
F_\testa(\rho) & = F_\testa (\ptr{2}{\Lambda(\rho)}) \, , \label{eq:reform-broadcast-1} \\
F_\testb(\rho) & = F_\testb (\ptr{1}{\Lambda(\rho)}) \label{eq:reform-broadcast-2}
\end{align}
for all $\rho \in \tests$.
\end{proposition}

\begin{proof}
Assume that a channel $\Lambda$ passes the broadcasting test $(\tests, \testa, \testb)$. For any $\A \in \testa$ and $\rho \in \tests$ we must have $\A(\rho) = \A(\ptr{2}{\Lambda(\rho)})$ from which it follows that $\rho \approx_\testa \ptr{2}{\Lambda(\rho)}$. Therefore, we get $F_\testa (\rho) = F_\testa (\ptr{2}{\Lambda(\rho)})$ and so \eqref{eq:reform-broadcast-1} holds. In a similar fashion we get $\rho \approx_\testb \ptr{1}{\Lambda(\rho)}$ and so also \eqref{eq:reform-broadcast-2} holds.

Let us then assume that \eqref{eq:reform-broadcast-1} and \eqref{eq:reform-broadcast-2} hold. For $\A \in \testa$ and $\rho \in \tests$ we have
\begin{equation}
\A(\rho) = (\A' \circ F_\testa )(\rho) = (\A' \circ F_\testa)(\ptr{2}{\Lambda(\rho)}) = \A(\ptr{2}{\Lambda(\rho)})
\end{equation}
where we have used \eqref{eq:reform-factor-obseravable} twice. For any $\B \in \testb$ we have
\begin{equation}
\B(\rho) = (\B' \circ F_\testb )(\rho) = (\B' \circ F_\testb (\ptr{1}{\Lambda(\rho)} = \B (\ptr{1}{\Lambda(\rho)}
\end{equation}
and thereby $\Lambda$ passes the broadcasting test $(\tests, \testa, \testb)$.
\end{proof}

The previously expressed statement that a broadcasting test depends only on the corresponding factorization maps and not on the specific details of measurements is a simple but useful fact. In the following we demonstrate the consequences of Prop.~\ref{prop:reform-broadcast}.

\begin{figure}
\centering
\includegraphics[height=4cm]{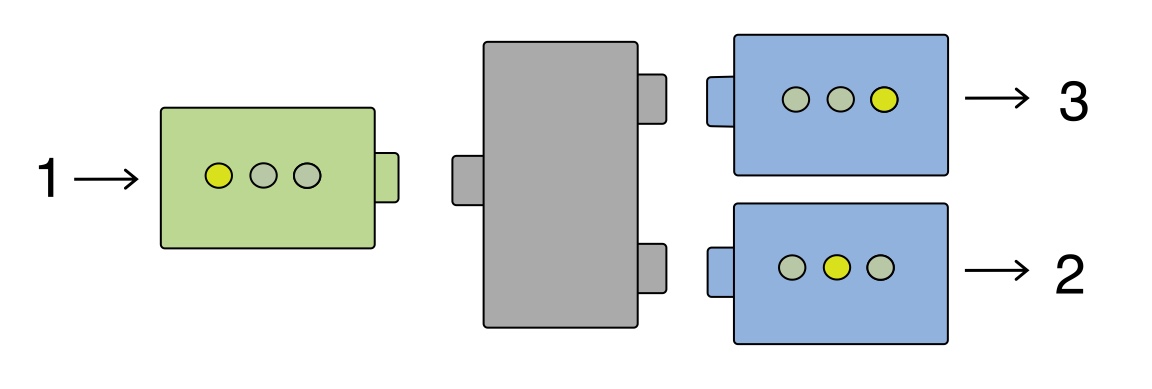}
\caption{ \label{fig:anti} A uniform antidiscrimination setup requires that the obtained measurement outcome is not the sent input put some other index, obtained with uniform probability. Even if such a setup would exist, the setup may not be broadcastable. This is the case e.g. for three inputs and a qubit system as an information carrier.}
\end{figure}

\begin{example}(\emph{Uniform antidiscrimination of qubit states}.)
Antidiscrimination of states $\rho_1,\ldots,\rho_n$ means that there is a $n$-outcome measurement that gives always a different outcome than what is encoded in a state, i.e., $\tr{\rho_x \A_x}=0$ for all $x=1,\ldots,n$. A stronger form is uniform antidiscrimination, which further requires that the other outcomes are obtained with uniform probability, i.e.,
\begin{equation}\label{eq:ua}
\tr{\rho_x \A_y}=\frac{1-\delta_{xy}}{n-1}
\end{equation}
for all $x,y=1,\ldots,n$. In the case of a qubit system, uniform antidiscrimination is possible for $n=2,3,4$ but not for higher $n$ \cite{HeKe19}. We can then ask if there is a channel that broadcasts these antidiscrimination setups, meaning that we look at a broadcasting test with $\tests=\{\rho_1,\ldots,\rho_n\}$, $\testa = \testb = \{\A\}$ and \eqref{eq:ua} is satisfied (see Fig.~\ref{fig:anti}). We consider the cases $n=2,3,4$ separately.

For $n=2$ uniform antidiscrimination is implemented with two perfectly distinguishable states and $\A$ is the measurement that discriminates them, just we relabel to outcomes to get antidiscrimination. As a set of perfectly distinguishable states is broadcastable, the task is possible.

For $n=4$ uniform antidiscrimination is implemented with four states that have Bloch vectors pointing to the corners of a regular tetrahedron. The measurement $\A$ has rank-1 effects that point to the opposite directions. We conclude that $\A$ is informationally complete, therefore there is no channel that passes the broadcasting test as the states do not belong to a simplex generated by jointly distinguishable states.

The case $n=3$ is the most interesting as it is not as evident as the previous two cases. Uniform antidiscrimination is implemented with three states that have Bloch vectors pointing to the corners of a equilateral triangle. The measurement $\A$ has rank-1 effects that point to the opposite directions. Without loosing generality we can assume that the states and effects are in the plane spanned by $\sigma_x$ and $\sigma_y$. As in Example~\ref{ex:factor-qubit} we identify $\state_\A$ with the disk $\{ (x,y) : x^2+y^2 \leq 1\}$ and the factorization map $F_{\A}$ is $F_{\A}(\rho) = (\tr{\rho \sigma_x},\tr{\rho \sigma_y})$. We conclude that a channel that passes the broadcasting test would implement the hypothetical perfect broadcasting device of the disk state space, which is a contradiction as only classical state spaces admit broadcasting \cite{BaBaLeWi07}. Therefore, the test is not broadcastable.
\end{example}

By using factorization maps we get an easy proof for the fact that the smallest amount of noise that one needs to add in order to make the scenario broadcastable is either 0 or 1. Thus adding white noise to non-broadcastable measurements does not make them broadcastable. This was already observed in slightly different forms in \cite{Heinosaari16,HeLePl19,mitra2021layers}.

\begin{proposition}\label{prop:noise}
Let $\testa$ and $\testb$ be subsets of measurements and define $\testa_{s}$ and $\testb_{t}$ to be the noisy versions of $\testa$ and $\testb$:
\begin{align}
\testa_{s} & = \{ (1-s) \A + s \C : \A \in \testa, \C \in \coin \}, \\
\testb_{t} & = \{ (1-t) \B + t \C : \B \in \testa, \C \in \coin \},
\end{align}
where $s, t \in (0,1]$ are fixed noise parameters and $\coin$ denotes the set of coin-toss measurements, i.e., POVMs such that every element is proportional to the identity operator $\id$. A channel $\Lambda: \state \to \state \otimes \state$ passes a broadcasting test $(\tests, \testa, \testb)$ if and only if it passes a broadcasting test $(\tests, \testa_{s}, \testb_{t})$.
\end{proposition}

\begin{proof}
For any $s \neq 0$, the equivalence classes determined by $\testa$ and $\testa_{s}$ are the same. This means that the respective factorization maps are the same, hence the claim follows from Prop.~\ref{prop:reform-broadcast}.
\end{proof}

We emphasize that Prop.~\ref{prop:noise} refers to the situation where noisy measurements are measured and the measurement outcome distributions are compared to the expected distributions. In contrast, the optimal quantum cloning device \cite{Werner98,KeWe99} can be seen as performing universal broadcasting at the cost of adding noise to all measurements (see e.g. \cite{HeScToZi14}).

From now on, we will mostly write broadcasting tests in terms of the corresponding factorization maps. This will simplify our calculations as it is easier to consider two maps instead of two subsets of measurements. Moreover, we will use \eqref{eq:reform-broadcast-1} and \eqref{eq:reform-broadcast-2} to generalize the idea of broadcasting tests also to broadcasting of channels.

\section{Broadcasting of channels}

There is a slightly different viewpoint to broadcasting that allows a more flexible starting point. From Prop.~\ref{prop:reform-broadcast} we conclude that the broadcasting condition demands that $\Lambda$ copies the test state but is 'invisible' for factorization maps. Instead of factorization maps, we can insert any single-system channels into that condition and this leads to the following definition.

\begin{definition} \label{def:broadcast-def}
Let $\Phi_i: \state \to \state_i$, $i \in \{1, 2\}$, be channels and let $\tests \subseteq \state$ be a convex set.
We say that the triple $(\tests, \Phi_1, \Phi_2)$ is broadcastable if there exists a channel $\Lambda: \state \to \state \otimes \state$ such that for all $\rho \in \tests$ we have
\begin{align}
\Phi_1 (\rho) & = \Phi_1(\ptr{2}{\Lambda(\rho)}) \, , \label{eq:broad-chan-1} \\
\Phi_2 (\rho) & = \Phi_2(\ptr{1}{\Lambda(\rho)}) \, . \label{eq:broad-chan-2}
\end{align}
In this case we also say that $\Phi_1$ and $\Phi_2$ are $\tests$-broadcastable.
\end{definition}

The conditions \eqref{eq:broad-chan-1}--\eqref{eq:broad-chan-2} can be again seen as a sort of broadcasting test; if we are acting with $\Phi_1$ and $\Phi_2$ on the duplicates, then it seems like $\Lambda$ is a perfect broadcasting device.

Definition~\ref{def:broadcast-def} resembles the definition of compatibility of channels \cite{Haapasalo15,HeMi17}, but has a crucial difference. Namely, we recall that two channels $\Phi_i: \state \to \state_i$, $i \in \{1, 2\}$, are compatible if there exists a channel $\Lambda: \state \to \state_1 \otimes \state_2$ such that for all $\rho\in\state$, we have
\begin{align}
\Phi_1 (\rho) & = \ptr{2}{\Lambda(\rho)} \, , \label{eq:comp-chan-1} \\
\Phi_2 (\rho) & = \ptr{1}{\Lambda(\rho)} \, . \label{eq:comp-chan-2}
\end{align}
Further, $\Phi_1$ and $\Phi_2$ are called $\tests$-compatible if this condition holds for all $\rho\in\tests \subset \state$ \cite{heinosaari2021testing}.
By comparing \eqref{eq:broad-chan-1}--\eqref{eq:broad-chan-2} to \eqref{eq:comp-chan-1}--\eqref{eq:comp-chan-2} we see that the latter requires that $\Lambda$ can reproduce the actions of $\Phi_1$ and $\Phi_2$, while in the former $\Lambda$ should duplicate the input state and let $\Phi_1$ and $\Phi_2$ act on the duplicates. In the following we show that broadcastability is a stronger requirement than compatibility, although for some classes of channels they are equivalent properties.

\begin{proposition}
Let $\Phi_i: \state \to \state_i$, $i \in \{1, 2\}$, be channels and let $\tests \subseteq \state$ be a convex set.
\begin{itemize}
\item[(a)] If $\Phi_1$ and $\Phi_2$ are $\tests$-broadcastable, then $\Phi_1$ and $\Phi_2$ are $\tests$-compatible.
\item[(b)] Assume that $\Phi_1$ and $\Phi_2$ are idempotent channels (i.e. $\state_1=\state_2=\state$ and $\Phi_1^2=\Phi_1$, $\Phi_2^2=\Phi_2$). Then $\Phi_1$ and $\Phi_2$ are $\tests$-broadcastable if and only if $\Phi_1$ and $\Phi_2$ are $\tests$-compatible.
\end{itemize}
\end{proposition}

\begin{proof}
\begin{itemize}
\item[(a)] Assume that $\Lambda$ satisfies \eqref{eq:broad-chan-1}--\eqref{eq:broad-chan-2}. We define another channel $\Lambda'$ as $\Lambda'=(\Phi_1 \otimes \Phi_2)\circ \Lambda$. Then $\Lambda'$ (in the place of $\Lambda$) satisfies \eqref{eq:comp-chan-1}--\eqref{eq:comp-chan-2}.
\item[(b)] We need to prove the 'if' part as the other direction holds generally. Assume that $\Lambda$ satisfies \eqref{eq:comp-chan-1}--\eqref{eq:comp-chan-2}, then $\Phi_1(\ptr{2}{\Lambda(\rho)}) = \Phi_1(\Phi_1(\rho)) = \Phi_1(\rho)$. It follows that $\Lambda$ satisfies \eqref{eq:broad-chan-1}--\eqref{eq:broad-chan-2}.
\end{itemize}
\end{proof}

In the following example we demonstrate that broadcastability is a strictly stronger relation for channels than compatibility.

\begin{example}(\emph{Invertible channels are not broadcastable but can be compatible})
Assume $\Phi_1$ and $\Phi_2$ are channels that are invertible, but we do not require the inverses to be positive. If $\Phi_1$ and $\Phi_2$ are broadcastable, we can apply the inverse maps to \eqref{eq:broad-chan-1} and \eqref{eq:broad-chan-2} and get $\rho = \ptr{2}{\Lambda(\rho)} = \ptr{1}{\Lambda(\rho)}$. It follows that $\Lambda$ is a universal broadcasting device, which contradicts the no-broadcasting theorem. We conclude that a pair of invertible channels is not broadcastable. However, there are compatible pairs of invertible channels. For instance, a partially dephasing channel $\Delta_\mu(\rho) = \mu \rho + (1-\mu) \frac{\id}{\dim(\Ha)}$ is invertible for all parameters $0<\mu \leq 1$. Two such channels are compatible for suitably small parameters \cite{Haapasalo19,girard2021jordan}.
\end{example}

Factorization maps are special kind of channels. Hence, in the terminology of Def.~\ref{def:broadcast-def} the result of Prop.~\ref{prop:reform-broadcast} becomes:

\begin{corollary}
A broadcasting test $(\tests, \testa, \testb)$ is broadcastable if and only if the triple $(\tests, F_\testa, F_\testb)$ is broadcastable.
\end{corollary}

As for subsets of measurements, we can assign factorization maps for subsets of channels. For our purposes, it is enough to have the following definition for single channels although an analogous definition can be phrased for subset of channels. For $\rho, \sigma \in \state$, we define $\rho \approx_\Phi \sigma$ if $\Phi(\rho) = \Phi(\sigma)$. Let $\state_{\Phi}$ be the set of equivalence classes of $\approx_\Phi$ and let $F_\Phi: \state \to \state_{\Phi}$ be the corresponding factorization map. We introduce a map $\Phi'$ on $\state_{\Phi}$ as
\begin{equation}
\Phi'( [\rho] ) = \Phi(\rho)
\end{equation}
for all $\rho \in \state$. This is again well-defined since by construction $\Phi$ is constant on the equivalence classes of $\approx_\Phi$.
We therefore have
\begin{equation}
\Phi = \Phi' \circ F_\Phi
\end{equation}
and the channel $\Phi$ hence factorizes through $F_\Phi$.
We note that $F_\Phi$ is the identity map for injective channels $\Phi$.

We will next show that for broadcasting tests one can restrict to the case of broadcastable factorization channels.

\begin{proposition}
A triple $(\tests, \Phi_1, \Phi_2)$ is broadcastable if and only if the triple $(\tests, F_{\Phi_1}, F_{\Phi_2})$ is broadcastable.
\end{proposition}

\begin{proof}
The proof is almost the same as proof of Prop.~\ref{prop:reform-broadcast}. Let $\Lambda: \state \to \state \otimes \state$ be the the channel that passes the broadcasting test $(\tests, \Phi_1, \Phi_2)$. Let $\rho \in \tests$. We have
\begin{equation}
\Phi_1 (\rho) = \Phi_1(\ptr{2}{\Lambda(\rho)})
\end{equation}
from which it follows that
\begin{equation}
\rho \approx_{\Phi_1} \ptr{2}{\Lambda(\rho)}
\end{equation}
and so we must have
\begin{equation}
F_{\Phi_1}(\rho) = F_{\Phi_1}(\ptr{2}{\Lambda(\rho)}) \, ,
\end{equation}
which shows that $\Lambda$ passes the broadcasting test $(\tests, F_{\Phi_1}, \Phi_2)$.

Let us then assume that $\Lambda$ passes the broadcasting test $(\tests, F_{\Phi_1}, \Phi_2)$ and let $\rho \in \tests$. We have
\begin{equation}
\Phi_1(\rho) = \Phi'_1 ( F_{\Phi_1}(\rho) ) = \Phi'_1(F_{\Phi_1}(\ptr{2}{\Lambda(\rho)})) = \Phi_1(\ptr{2}{\Lambda(\rho)}) \, ,
\end{equation}
which shows that $\Lambda$ passes the broadcasting test $(\tests, \Phi_1, \Phi_2)$.
\end{proof}

One can also obtain the same result as above for broadcasting tests with two subsets of channels, similarly as we did earlier with two subsets of measurements. The results and their proofs would be the same.

Moreover one can also boil down the broadcasting of channels to broadcasting of measurements:
\begin{theorem}
Let $\Phi_i: \state \to \state_i$, $i \in \{1,2\}$ be channels and define
\begin{align}
\testa & = \{ \A \circ \Phi_1: \A \in \obs(\state_1) \}, \\
\testb & = \{ \B \circ \Phi_2: \B \in \obs(\state_2) \} \, .
\end{align}
Then $(\tests, \Phi_1, \Phi_2)$ is broadcastable if and only if $(\tests, \testa, \testb)$ is broadcastable.
\end{theorem}

\begin{proof}
The proof follows from $F_{\Phi_1} = F_\testa$ and $F_{\Phi_2} = F_\testb$. We prove the first equality, the second one follows in the same way. Assume that $\Phi_1(\rho) \neq \Phi_1(\sigma)$. Then there is a measurement $\A \in \obs(\state_1)$ such that $\A(\Phi_1(\rho)) \neq \A(\Phi_1(\sigma))$ and so $\A'(\rho) \neq \A'(\sigma)$ for $\A' = \A \circ \Phi_1$. Conversely, assume that $\A'(\rho) \neq \A'(\sigma)$ for some $\A' \in \testa$. Then $\A' = \A \circ \Phi_1$ for some $\A \in \obs(\state_1)$ and we have $\A(\Phi_1(\rho)) \neq \A(\Phi_1(\sigma))$. It follows that we must have $\Phi_1(\rho) \neq \Phi_1(\sigma)$.
\end{proof}

\section{Special cases}\label{sec:special}

In this section we are going to look closely at special cases of broadcasting of measurements. We have already seen that all other introduced scenarios can be reduced to this one. We will first show that commutativity of the measurements or some measurements and states is sufficient for broadcastability.
\begin{proposition} \label{prop:special-commuteAB}
If $\testa$ and $\testb$ are mutually commuting, then a broadcasting test $(\tests, \testa, \testb)$ is broadcastable.
\end{proposition}
\begin{proof}
If $\testa$ and $\testb$ are mutually commuting, then there is an isomorphism $U: \Ha \to \oplus_{k=1}^n (\Ha^\testa_k \otimes \Ha^\testb_k)$ such that $U \A_i U^* = \oplus_{k=1}^n \A'_{i|k} \otimes \id_\testb$ and $U \B_j U^* = \oplus_{k=1}^n \id_\testa \otimes \B'_{j|k}$ \cite[Theorem 1]{ScholzWerner-Tsirelson}, where $\oplus$ denotes the direct sum of Hilbert spaces and operators respectively. The result easily follows, since measurements of the form $\A'_{i|k} \otimes \id_\testb$ and $\id_\testa \otimes \B'_{j|k}$ act on separate components and are hence straightforward to broadcast.
\end{proof}

Commutativity is not a necessary condition, we will provide counter-examples of broadcastable broadcasting tests of the form $(\state, \testa, \testa)$ in Exm.~\ref{ex:saa-nonCommutative} and $(\state, \testa, \testb)$ in Exm.~\ref{ex:sab-nonCommutative} where the respective measurements do not commute.

\begin{proposition} \label{prop:special-commuteAT}
If $\testa$ and $\tests$ are mutually commuting, then a broadcasting test $(\tests, \testa, \testb)$ is broadcastable.
\end{proposition}
\begin{proof}
We can again use \cite[Theorem 1.]{ScholzWerner-Tsirelson} to show that $\Ha$ is
isomorphic to $\oplus_{k=1}^n \Ha^\testa_k \otimes \Ha^\tests_k$ and, up to an
isomorphism, we have $\A_i = \oplus_{k=1}^n \A'_{i|k} \otimes \id_\tests$ and $\rho =
\oplus_{k=1}^n p_k \frac{\id_\testa}{\dim(\Ha_k^\testa)} \otimes \rho_k'$, where now $p_k \geq 0$ and $\sum_{k=1}^n p_k = 1$ are just to ensure proper normalization of $\rho$. The result follows immediately, since $\A$ measures only the classical information encoded in the probability distribution $p_k$.
\end{proof}

\subsection{The case when $\state \subset \lin{\tests}$}
As the first special case we will assume that $\state \subset \lin{\tests}$, which is a generalization of the case when $\tests = \state$. If all of the measurements in $\testa$ commute, then according to Prop.~\ref{prop:special-commuteAB} the test is broadcastable. This is the case when all of the measurements in $\testa$ can be post-processed from a single projective measurement. To be more specific, let $\Ha_3$ be the qutrit Hilbert space, $\dim(\Ha_3) = 3$, with the orthonormal basis $\{\ket{i}\}_{i=1}^3$, and let $\testa = \{ \A \}$ such that $A_i = \ketbra{i}$. Then $(\tests, \testa, \testa)$ is broadcastable, the channel that passes the broadcasting test is defined as
\begin{equation} \label{eq:saa-dim3Lambda}
\Lambda(\rho) = \sum_{i=1}^3 \tr{\rho \A_i} \ketbra{ii} \, .
\end{equation}
To see that $\Lambda$ passes the broadcasting test, simply observe that
\begin{equation}
\ptr{1}{\Lambda(\rho)} = \ptr{2}{\Lambda(\rho)} = \sum_{i=1}^3 \tr{\rho \A_i} \ketbra{i}
\end{equation}
and so measuring $\A$ will yield the correct probability. In the following example we will show that commutativity is only sufficient but not necessary condition for broadcastability; the main idea is that if we extend the Hilbert space, then we have some freedom in the choice of $\A_i$ in \eqref{eq:saa-dim3Lambda}.

\begin{example}[\emph{Broadcastable but noncommuting measurements in the $(\state, \testa, \testa)$ scenario}] \label{ex:saa-nonCommutative}
Consider the Hilbert space $\Ha_5$, $\dim(\Ha_5) = 5$, with the orthonormal basis $\{\ket{i}\}_{i=1}^5$. Let $\tests = \state$ and $\testa = \{\A\}$ defined for $i \in \{1,2,3\}$ as
\begin{equation}
\A_i = \ketbra{i} + \M_i
\end{equation}
where $\M$ is an arbitrary measurement supported on $\lin{\{\ket{4}, \ket{5}\}}$, that is, $\M_i \geq 0$ for all $i \in \{1,2,3\}$, and $\sum_{i=1}^3 \M_i = \ketbra{4} + \ketbra{5}$.
This implies that $\M_i \ket{j} = 0$ for all $i \in \{1,2,3\}$ and $j \in \{1,2,3\}$. One can choose $\M_i$ such that they do not mutually commute, which then implies that also $\A_i$ do not commute. To see that $(\tests, \testa, \testa)$ is broadcastable, simply consider channel defined analogically as in \eqref{eq:saa-dim3Lambda}, that is $\Lambda(\rho) = \sum_{i=1}^3 \tr{\rho \A_i} \ketbra{ii}$. We again have
\begin{equation}
\ptr{1}{\Lambda(\rho)} = \ptr{2}{\Lambda(\rho)} = \sum_{i=1}^3 \tr{\rho \A_i} \ketbra{i}
\end{equation}
and the result follows from $\tr{\ketbra{i} \A_j} = \delta_{ij}$ for all $i \in \{1,2,3\}$ and $j \in \{1,2,3\}$.
\end{example}

\begin{restatable}{theorem}{thmSaa} \label{thm:saa}
Assume that $\state \subset \lin{\tests}$. Then $(\tests, \testa, \testa)$ is broadcastable if and only if there is a single measurement $\G$ such that the norm of every non-zero operator in $\G$ is $1$, i.e., $\no{\G_\lambda} \in \{0, 1\}$, $\forall \lambda$ and every measurement in $\A$ can be post-processed from $\G$, i.e., for every $\A \in \testa$ and every $i$ we have
\begin{equation}
\A_i = \sum_\lambda p(i|\A, \lambda) \G_\lambda
\end{equation}
where $p(i|\A, \lambda) \geq 0$ and $\sum_i p(i|\A, \lambda) = 1$.
\end{restatable}
\begin{proof}
To see that this is sufficient, we only need to show that $(\state, \{\G\}, \{\G\})$ passes the broadcasting test. If $\no{\G_\lambda} = 1$, then there is unit vector $\ket{\lambda} \in \Ha$ such that $\G_\lambda \ket{\lambda} = \ket{\lambda}$. Since $\sum_\lambda \G_\lambda = \id$, we then must have $\bra{\lambda} \G_{\mu} \ket{\lambda} = \delta_{\lambda \mu}$, where $\lambda$ and $\mu$ both index the outcomes of $\G$. We then construct the channel
\begin{equation}
\Lambda(\rho) = \sum_\lambda \tr{\rho \G_\lambda} \ketbra{\lambda \lambda}.
\end{equation}
We again have $\ptr{1}{\Lambda(\rho)} = \ptr{2}{\Lambda(\rho)} = \sum_\lambda \tr{\rho \G_\lambda} \ketbra{\lambda}$ and the result follows. The proof in the other direction is in the Appendix~\ref{appendix:proofSaa}.
\end{proof}

\begin{corollary}
Assume that $\state \subset \lin{\tests}$ and that $\dim(\Ha) \leq 4$. Then $(\tests, \testa, \testa)$ is broadcastable if and only if for every $\A, \B \in \testa$ we have that $\A_i$ and $\B_j$ commute for all $i,j$.
\end{corollary}
\begin{proof}
Assume that $(\tests, \testa, \testa)$ is broadcastable, the result follows by showing that the measurement $\G_\lambda$ as described in Thm.~\ref{thm:saa} can be post-processed from a projective measurement. It is then straightforward to show that $\A_i$ and $\B_j$ commute for every $\A, \B \in \testa$ and for all $i,j$.

If $\G$ is dichotomic, i.e., it has only two outcomes, then it is commutative, so we need to consider only the case when $\G$ is at least trichotomic, i.e., it has at least three nonzero outcomes. Since then for these three outcomes we must have $\no{\G_\lambda} = 1$, it follows that we must have at least $\dim(\Ha) = 3$, in which case $\G$ is projective measurement. Thus the only case we need to investigate is when $\dim(\Ha) = 4$ and $\G$ has three nonzero outcomes, since with four nonzero outcomes $\G$ is again projective measurement. So let $\dim(\Ha) = 4$ and let $\{\ket{i}\}_{i=1}^4$ be a suitable orthonormal basis. Then $\G_i = \ketbra{i} + \M_i$ for $i \in \{1,2,3\}$, which implies that $\M_i = \alpha_i \ketbra{4}$ for all $i \in \{1,2,3\}$ and thus again $\G$ can be post-processed from a single projective measurement.
\end{proof}

In the general case when $\testa \neq \testb$ commutativity is also a sufficient but not necessary criterion for broadcasting. This is demonstrated by the following example, where we use $B(\Ha)$ to denote the set of operators on the Hilbert space $\Ha$.
\begin{example}[\emph{Broadcastable but noncommuting measurements in the $(\state, \testa, \testb)$ scenario}] \label{ex:sab-nonCommutative}
Let $\dim(\Ha_6) = 6$, let $\{\ket{i}\}_{i=1}^6$ be an orthonormal basis and let $P = \sum_{i=1}^4 \ketbra{i}$. Then $P\Ha$ is a four-dimensional Hilbert space, thus it can be seen as tensor product of two qubit Hilbert spaces $\Ha_2$, $\dim(\Ha_2) = 2$, i.e., $P\Ha = \Ha_2 \otimes \Ha_2$. Moreover let $\iota: B(\Ha_2) \to B(\Ha_6)$ be the embedding defined by $\iota(\ket{i} \! \bra{j}) = \ket{i+5} \! \bra{j+5}$ for $i,j \in \{0,1\}$, that is $\iota$ maps operators on $\Ha_2$ to operators on $P^\perp B(\Ha_6)$.

Let $\A, \B$ be projective measurements on $\Ha_2$ given as
\begin{align}
 & \A_0 = \ketbra{0}, & & \A_1 = \ketbra{1}, \\
 & \B_0 = \ketbra{+}, & & \B_1 = \ketbra{-},
\end{align}
where $\{\ket{i}\}_{i=0}^1$ is an orthonormal basis of $\Ha_2$ and $\ket{\pm} = \frac{1}{\sqrt{2}}(\ket{0} \pm \ket{1})$.

Let $\Delta_\mu: B(\Ha_2) \to B(\Ha_2)$ be the partially dephasing channel given as $\Delta_\mu(\rho) = \mu(\rho) + (1-\mu) \frac{\id}{2}$. It is known that for $\mu \leq \frac{1}{3}$ the channel $\Delta_\mu$ is self-compatible, i.e., there is a channel $\Psi: B(\Ha_2) \to B(\Ha_2 \otimes \Ha_2)$ such that $\Delta_\mu(\rho) = \ptr{1}{\Psi(\rho)} = \ptr{2}{\Psi(\rho)}$. We now define measurements $\tilde{\A}$, $\tilde{\B}$ on $\Ha_6$ such that they do not commute, but are broadcastable. Let
\begin{align}
\tilde{\A}_i & = \A_i \otimes \id_2 + \iota(\Delta_{\frac{1}{3}}(\A_i)), \\
\tilde{\B}_i & = \id_2 \otimes \B_i + \iota(\Delta_{\frac{1}{3}}(\B_i)),
\end{align}
where $\id_2$ is the identity on $\Ha_2$ and the terms $\A_i \otimes \id_2$ and $\id_2 \otimes \B_i$ are supported on $P \Ha_6 = \Ha_2 \otimes \Ha_2$. In order to construct the channel $\Lambda$ that passes the broadcasting test $(\tests, \testa, \testb)$ we will need the auxiliary channel $\Pi: B(P\Ha_6) \to B(P\Ha_6 \otimes P\Ha_6)$ defined for $X_1, X_2 \in B(\Ha_2)$ as $\Pi(X_1 \otimes X_2) = X_1 \otimes \frac{\id}{2} \otimes \frac{\id}{2} \otimes X_2$ and extended by linearity, here we are heavily relying on the identification $P\Ha_6 = \Ha_2 \otimes \Ha_2$. Then let
\begin{equation}
\Lambda(\rho) = (P \otimes P) (\Pi(P \rho P) + (\Pi \circ \Psi \circ \iota^{-1})(P^\perp \rho P^\perp)) (P \otimes P)
\end{equation}
where $\iota^{-1}: B(P^\perp \Ha_6) \to B(\Ha_2)$ is the inverse to $\iota$. $\Lambda$ clearly is completely positive, we only need to prove that it is trace-preserving, but this follows from $\tr{(\Pi(P \rho P)) (P \otimes P)} = \tr{\rho P}$ and $\tr{(\Pi \circ \Psi \circ \iota^{-1})(P^\perp \rho P^\perp) (P \otimes P)} = \tr{\rho P^\perp}$. Finally, to show that $\Lambda$ passes the broadcasting test $(\tests, \testa, \testb)$, we have
\begin{equation}
\begin{split}
\tr{\Lambda(\rho)(\tilde{\A}_i \otimes \id_6)} & = \tr{ (P \otimes P) (\Pi(P \rho P) + (\Pi \circ \Psi \circ \iota^{-1})(P^\perp \rho P^\perp)) (P \otimes P) (\tilde{\A}_i \otimes \id_6)} \\
& = \tr{ (\Pi(P \rho P) + (\Pi \circ \Psi \circ \iota^{-1})(P^\perp \rho P^\perp)) ((\A_i \otimes \id_2) \otimes P)} \\
& = \tr{ (P \rho P) (\A_i \otimes \id_2)}
+ \tr{ (\Psi \circ \iota^{-1})(P^\perp \rho P^\perp) (\A_i \otimes \id_2)} \\
& = \tr{ \rho (\A_i \otimes \id_2)}
+ \tr{ \rho (\iota \circ \Delta_{\frac{1}{3}})(\A_i) } \\
& = \tr{\rho \tilde{\A}_i } .
\end{split}
\end{equation}
the proof for $\tilde{\B}_j$ is similar.
\end{example}

\subsection{The case when $\testb$ is info-complete}

It is sufficient that $\testb$ is infomationally-complete. In this case we obtain the following result:
\begin{restatable}{theorem}{thmTao} \label{thm:tao}
Assume that $\testb$ is informationally-complete. If $(\tests, \testa, \testb)$ is broadcastable, then for every $\A \in \testa$ there is $\tilde{\A} \in \obs(\state)$ such that $\A(\rho) = \tilde{\A}(\rho)$ for all $\rho \in \tests$ and $\tilde{\A}_i$ commutes with every $\rho \in \tests$.
\end{restatable}

\begin{proof}
The proof of is in the Appendix~\ref{appendix:proofTao}.
\end{proof}

Finally we present a simple example that demonstrates why we sometimes need to use $\tilde{\A}$ instead of $\A$.
\begin{example}
Consider the qubit system and assume that $\tests$ is the set of density matrices diagonal in the computational basis and $\testa = \{\A\}$ where $\A$ is the projective measurement corresponding to the eigenvectors of $\sigma_x$. Clearly there are operators in $\tests$ and $\testa$ that do not commute. But note that $\tr{\rho \A_i} = \frac{1}{2}$ for all $\rho \in \tests$ and so for $\tilde{\A}_i = \frac{\1}{2}$ we have $\A(\rho) = \tilde{\A}(\rho)$ and $[\rho, \tilde{\A}_i] = 0$ for all $\rho \in \tests$.
\end{example}

\section{Conclusions}
The no-broadcasting theorem is among the most important limitations in quantum information processing. It has, however, very strong assumptions as it uses all states and all measurements in its usual formulation. In the current work we generalized the setting of accurate but limited broadcasting and introduce the concept of a broadcasting test. This enables to pose the question of the exact conditions when broadcasting is possible. Earlier studies have already revealed that commutativity, either between states or measurements, makes broadcasting possible and that is not surprising as the scenario is then effectively classical, but, as we have shown, commutativity is not necessary for broadcasting in general.

We have characterized the conditions for broadcastability in certain important special classes of broadcasting tests. While we did not cover the most general scenario, even the special cases that we have investigated turned out to provide nontrivial and unexpected results. We have also investigated generalizations of broadcasting to channels and we have shown that such generalizations can be mapped back to the case considering only states and measurements.

Future investigations into broadcasting can be done along several lines: one can, of course, attempt to solve the most general scenario, but this will likely be a hard task. Another option would be, in the spirit of \cite{BaBaLeWi07}, to generalize our results to more general systems, such as the ones described by operational probabilistic theories \cite{ChDaPe09pra} or general probabilistic theories \cite{muller2021probabilistic,Plavala21}. Yet another line of research would be to consider approximate broadcasting, that is allowed to return the correct result only with some probability, and generalize our result in this direction.

\begin{acknowledgments}
This research has been supported by Academy of Finland mobility cooperation funding (grant no 434228) and DAAD Joint Research Cooperation Scheme (project no 57570110).
A.J. acknowledges support from the grant VEGA 1/0142/20 and the Slovak Research and Development Agency grant APVV-20-0069.
M.P. acknowledges support from the Deutsche Forschungsgemeinschaft (DFG, German Research Foundation, project numbers 447948357 and 440958198), the Sino-German Center for Research Promotion (Project M-0294), the ERC (Consolidator Grant 683107/TempoQ), the German Ministry of Education and Research (Project QuKuK, BMBF Grant No. 16KIS1618K), and the Alexander von Humboldt Foundation.
\end{acknowledgments}

\bibliographystyle{unsrt}
\bibliography{bibliography}

\begin{thebibliography}{10}

\bibitem{QI01Werner}
R.F. Werner.
\newblock Quantum information theory -- an invitation.
\newblock In {\em Quantum Information: an Introduction to Basic Theoretical
  Concepts and Experiments}, chapter~2, pages 14--57. Springer-Verlag, 2001.

\bibitem{Barnumetal96}
H.~Barnum, C.M. Caves, C.A. Fuchs, R.~Jozsa, and B.~Schumacher.
\newblock Noncommuting mixed states cannot be broadcast.
\newblock {\em Phys. Rev. Lett.}, 76:2818--2821, 1996.

\bibitem{WoZu82}
W.K. Wootters and W.H. Zurek.
\newblock A single quantum cannot be cloned.
\newblock {\em Nature}, 299:802--803, 1982.

\bibitem{d2005superbroadcasting}
G.M. D'Ariano, C.~Macchiavello, and P.~Perinotti.
\newblock Superbroadcasting of mixed states.
\newblock {\em Physical review letters}, 95:060503, 2005.

\bibitem{buscemi2006universal}
F.~Buscemi, G.M. D'Ariano, C.~Macchiavello, and P.~Perinotti.
\newblock Universal and phase-covariant superbroadcasting for mixed qubit
  states.
\newblock {\em Physical Review A}, 74:042309, 2006.

\bibitem{ScIbGiAc05}
V.~Scarani, S.~Iblisdir, N.~Gisin, and A.~Acin.
\newblock Quantum cloning.
\newblock {\em Rev. Modern Phys.}, 77:1225--1256, 2005.

\bibitem{CeFi06}
N.J. Cerf and J.~Fiur\'{a}\v{s}ek.
\newblock Optimal quantum cloning - a review.
\newblock In {\em Progress in Optics}, volume~49. Elsevier, 2006.

\bibitem{BaBaLeWi07}
H.~Barnum, J.~Barrett, M.~Leifer, and A.~Wilce.
\newblock Generalized no-broadcasting theorem.
\newblock {\em Phys. Rev. Lett.}, 99:240501, 2007.

\bibitem{Heinosaari16}
T.~Heinosaari.
\newblock Simultaneous measurement of two quantum observables: Compatibility,
  broadcasting, and in-between.
\newblock {\em Phys. Rev. A}, 93:042118, 2016.

\bibitem{HeKe19}
T.~Heinosaari and O.~Kerppo.
\newblock Communication of partial ignorance with qubits.
\newblock {\em J. Phys. A: Math. Theor.}, 52:395301, 2019.

\bibitem{HeLePl19}
T.~Heinosaari, L.~Lepp\"aj\"arvi, and M.~Pl\'avala.
\newblock No-free-information principle in general probabilistic theories.
\newblock {\em Quantum}, 3:157, 2019.

\bibitem{mitra2021layers}
A.~Mitra.
\newblock Layers of classicality in the compatibility of measurements.
\newblock {\em Phys. Rev. A}, 104:022206, 2021.

\bibitem{Werner98}
R.F. Werner.
\newblock Optimal cloning of pure states.
\newblock {\em Phys. Rev. A}, 58:1827--1832, 1998.

\bibitem{KeWe99}
M.~Keyl and R.F. Werner.
\newblock Optimal cloning of pure states, testing single clones.
\newblock {\em J. Math. Phys.}, 40:546, 1999.

\bibitem{HeScToZi14}
T.~Heinosaari, J.~Schultz, A.~Toigo, and M.~Ziman.
\newblock Maximally incompatible quantum observables.
\newblock {\em Phys. Lett. A}, 378:1695--1699, 2014.

\bibitem{Haapasalo15}
E.~Haapasalo.
\newblock Robustness of incompatibility for quantum devices.
\newblock {\em J. Phys. A: Math. Theor.}, 48:255303, 2015.

\bibitem{HeMi17}
T.~Heinosaari and T.~Miyadera.
\newblock Incompatibility of quantum channels.
\newblock {\em J. Phys. A: Math. Theor.}, 50:135302, 2017.

\bibitem{heinosaari2021testing}
T.~Heinosaari, T.~Miyadera, and R.~Takakura.
\newblock Testing incompatibility of quantum devices with few states.
\newblock {\em Phys. Rev. A}, 104:032228, 2021.

\bibitem{Haapasalo19}
E.~Haapasalo.
\newblock Compatibility of covariant quantum channels with emphasis on {W}eyl
  symmetry.
\newblock {\em Ann. Henri Poincar\'e}, 20:3163, 2019.

\bibitem{girard2021jordan}
M.~Girard, M.~Pl\'avala, and J.~Sikora.
\newblock Jordan products of quantum channels and their compatibility.
\newblock {\em Nature Commun.}, 12(1):1--6, 2021.

\bibitem{ScholzWerner-Tsirelson}
V.B. Scholz and R.F. Werner.
\newblock Tsirelson's {P}roblem, 2008.

\bibitem{ChDaPe09pra}
G.~Chiribella, G.M. D'Ariano, and P.~Perinotti.
\newblock Theoretical framework for quantum networks.
\newblock {\em Phys. Rev. A}, 80:022339, 2009.

\bibitem{muller2021probabilistic}
Markus M{\"u}ller.
\newblock Probabilistic theories and reconstructions of quantum theory.
\newblock {\em SciPost Physics Lecture Notes}, page 028, 2021.

\bibitem{Plavala21}
M.~Pl\'avala.
\newblock General probabilistic theories: An introduction.
\newblock arXiv:2103.07469, 2021.

\bibitem{Munkres-topology}
J.R. Munkres.
\newblock {\em Topology}.
\newblock Prentice-Hall, 2000.

\bibitem{paulsen2002completely}
Vern Paulsen.
\newblock {\em Completely bounded maps and operator algebras}.
\newblock Cambridge University Press, Cambridge, 2002.

\bibitem{TOA1}
M.~Takesaki.
\newblock {\em Theory of operator algebras. {I}}, volume 124 of {\em
  Encyclopaedia of Mathematical Sciences}.
\newblock Springer-Verlag, Berlin, 2002.
\newblock Reprint of the first (1979) edition, Operator Algebras and
  Non-commutative Geometry, 5.

\end{thebibliography}

%
%

\numberwithin{lemma}{section}
\numberwithin{proposition}{section}
\numberwithin{theorem}{section}
\numberwithin{definition}{section}
\appendix

\section{Proof of Theorem~\ref*{thm:saa}} \label{appendix:proofSaa}
We denote by $B(\Ha)$ the set of operators on the Hilbert space $\Ha$ and by $B(\Ha)^+$ the set of positive semidefinite operators. We will assume that all of the Hilbert spaces are finite-dimensional.
\begin{definition}
Let $\Phi: \state \to \state$ be a channel. The adjoint channel (also called the Heisenberg picture) $\Phi^*: B(\Ha) \to B(\Ha)$ is given for $\rho \in \state$ and $E \in B(\Ha)$ as
\begin{equation}
\tr{\Phi^*(E) \rho} = \tr{E \Phi(\rho)} \, .
\end{equation}
\end{definition}
It is straightforward to show that $\Phi^*$ is a linear and completely positive map and that $\Phi^*$ is unital, i.e., $\Phi^*(\id) = \id$. The following lemma will play a key role in our calculations. The first one is heavily inspired by \cite[Lemma 1]{BaBaLeWi07}.
\begin{lemma} \label{lemma:saa-fixedPoint}
Let $\Phi: B(\Ha) \to B(\Ha)$ be a positive unital map and let
\begin{equation}
\fix(\Phi)^+ = \{ X \in B(\Ha)^+ : \Phi(X) = X \}
\end{equation}
be the set of positive fixed points of $\Phi$. There is a positive map
\begin{equation}
\Pi: B(\Ha) \to B(\Ha)
\end{equation}
such that for all $X \in B(\Ha)^+$, we have $\Pi(X) \in \fix(\Phi)^+$ and $\Pi$ is idempotent, i.e., $\Pi^2 = \Pi$. Moreover $\Pi$ is unital and:
\begin{enumerate}[(a)]
\item if $\Phi$ is completely positive, then $\Pi$ is completely positive,
\item if $\Phi$ is trace preserving, then $\Pi$ is trace preserving.
\end{enumerate}
\end{lemma}
\begin{proof}
Let
\begin{equation}
\Pi_n = \dfrac{1}{n} \sum_{k=1}^n \Phi^k,
\end{equation}
then the main idea of the proof is that $\Pi \approx \lim_{n \to \infty} \Pi_n$. In order
to make the limit meaningful we will show that the sequence $\{\Pi_n\}_{n=1}^\infty$ contains a convergent subsequence.

First of all observe that $\id \in \fix(\Phi)$ since $\Phi$ is unital. Define
\begin{equation}
B_1 = \{ X \in B(\Ha): -\id \leq X \leq \id \}
\end{equation}
and so for every $X \in B_1$ we have $0 \leq \frac{1}{2}(\id + X) \leq 1$, or, equivalently, for any $X \in B_1$ there is $E \in B(\Ha)$, $0 \leq E \leq \id$ such that $X = 2E - \id$. Note that $0 \leq \Phi(E)$ follows from positivity of $\Phi$ and $\Phi(E) \leq \id$ follows from $0 \leq \id - \Phi(E) = \Phi(\id - E)$ and again from positivity of $\Phi$. We thus get
\begin{equation} \label{eq:saa-fixedPoint-PhiXinB1}
\Phi(X) = 2 \Phi(E) - \id \in B_1.
\end{equation}
Since for every $X \in B_1$ we have $\no{X} \leq 1$ and for every hermitian $Y \in B(\Ha)$ such that $\no{Y} \leq 1$ we have $Y \in B_1$, we can express the superoperator norm of $\Phi$ as
\begin{equation}
\no{\Phi} = \inf \{ \mu \in \real : \no{\Phi(X)} \leq \mu, \forall X \in B_1 \}
\end{equation}
and using \eqref{eq:saa-fixedPoint-PhiXinB1} we get $\no{\Phi} \leq 1$. We then have
\begin{equation}
\no{\Pi_n} \leq \dfrac{1}{n} \sum_{k=1}^n \no{\Phi}^k \leq 1
\end{equation}
which shows that the sequence $\{ \Pi_n \}_{n=1}^\infty$ is a subset of the unit ball of superoperators. It follows that there is a convergent subsequence $\{ \Pi_{n_j} \}_{j=1}^\infty$, see e.g. \cite[Theorem 28.2]{Munkres-topology}. Denote
\begin{equation}
\Pi = \lim_{j \to \infty} \Pi_{n_j}.
\end{equation}
Let us first of all show that $\Pi$ is positive. Let $X \in \lhp$, then
\begin{equation}
\Pi(X) = \lim_{j \to \infty} \dfrac{1}{n_j} \sum_{k=1}^{n_j} \Phi^k(X).
\end{equation}
Since $\sum_{k=1}^{n_j} \Phi^k(X) \geq 0$, $\Pi(X)$ is the limit of positive operators. But then $\Pi(X) \geq 0$ since the positive cone $B(\Ha)^+$ is closed.

In order to show that $\Phi \circ \Pi = \Pi$, which is equivalent to $\Pi(X) \in \fix(\Phi)$ for $X \in B(\Ha)^+$, we will need to show that $\{ \Pi_{n_j + 1} \}_{j=1}^\infty$ converges to the same limit as $\{ \Pi_{n_j} \}_{j=1}^\infty$. First note that
\begin{equation} \label{eq:appFixedPoint-n+1}
\no{\Pi_{n+1} - \Pi_n} = \dfrac{1}{n+1} \no{ \Phi^{n+1} - \dfrac{1}{n} \sum_{k=1}^n \Phi^k } \leq \dfrac{2}{n+1}.
\end{equation}
Now we have
\begin{equation}
\no{ \Pi - \Pi_{n_j + 1} } \leq \no{ \Pi - \Pi_{n_j} } + \no{ \Pi_{n_j} - \Pi_{n_j + 1} }
\end{equation}
where we can use \eqref{eq:appFixedPoint-n+1} and that $\{ \Pi_{n_j} \}_{j=1}^\infty$ converges to $\Pi$ to show that $\lim_{j \to \infty} \Pi_{n_j + 1} = \Pi$. We now have
\begin{equation}
\begin{split}
\Phi(\Pi(X)) & = \lim_{j \to \infty} \dfrac{1}{n_j} \sum_{k=1}^{n_j} \Phi^{k+1}(X) \\
& = \lim_{j \to \infty} \left( \dfrac{1}{n_j} \sum_{k=1}^{n_j+1} \Phi^{k}(X) - \dfrac{1}{n_j} \Phi(X) \right) \\
& = \lim_{j \to \infty} \left( \dfrac{n_j + 1}{n_j} \right) \lim_{j \to \infty} \left( \dfrac{1}{n_j + 1} \sum_{k=1}^{n_j+1} \Phi^{k}(X) \right) - \lim_{j \to \infty} \left( \dfrac{1}{n_j} \Phi(X) \right) \\
& = \lim_{j \to \infty} \Pi_{n_j + 1}(X) = \Pi(X)
\end{split}
\end{equation}
and so $\Pi(X) \in \fix(\Phi)$ for all $X \in B(\Ha)^+$.

To show that $\Pi$ is idempotent we only need to show $\Pi(Y) = Y$ for all $Y \in \fix(\Phi)$, then $\Pi^2 = \Pi$ follows from $\Pi(X) \in \fix(\Phi)$. We have
\begin{equation}
\Pi(Y) = \lim_{j \to \infty} \dfrac{1}{n_j} \sum_{k=1}^{n_j} \Phi^k(Y) = \lim_{j \to \infty} \dfrac{1}{n_j} \sum_{k=1}^{n_j} Y = Y.
\end{equation}
Since $\Phi$ is unital, we have that $\id \in \fix(\Phi)$ and so $\Pi$ is unital as well.

Finally if $\Phi$ is completely positive, then also $\Phi^k$ is completely positive, which implies that $\Pi_n$ is completely positive and so $\Pi$ is a limit of completely positive maps. It then follows that also $\Pi$ is completely positive.

If $\Phi$ is trace preserving, then
\begin{equation}
\tr{\Pi(X)} = \lim_{j \to \infty} \dfrac{1}{n_j} \sum_{k=1}^{n_j} \tr{\Phi^k(X)} = \lim_{j \to \infty} \dfrac{1}{n_j} \sum_{k=1}^{n_j} \tr{X} = \tr{X}
\end{equation}
and so $\Pi$ is also trace preserving.
\end{proof}

\thmSaa*
\begin{proof}
Let $\Lambda$ be the channel that passes the broadcasting test, then we have $\Lambda^*(\A_i \otimes \id) = \Lambda^*(\id \otimes \A_i) = \A_i$. We can now, without the loss of generality, assume that $\swap \circ \Lambda = \Lambda$, for example by replacing $\Lambda$ with $\frac{1}{2}(\Lambda + \swap \circ \Lambda)$, which also passes the same broadcasting test whenever $\Lambda$ does. Here $\swap$ is the defined as $\swap(\rho_1 \otimes \rho_2) = \rho_2 \otimes \rho_1$. We then have that $\ptr{1}{\Lambda(\rho)} = \ptr{2}{\Lambda(\rho)}$, thus we define a channel $\Phi: \state \to \state$ as $\Phi(\rho) = \ptr{1}{\Lambda(\rho)}$.

Since $\Phi$ is trace preserving, we have that $\Phi^*$ is unital and thus, according to Lemma~\ref{lemma:saa-fixedPoint}, there is a projection $\Pi: B(\Ha) \to \fix(\Phi^*)$, where $\fix(\Phi^*) = \{ X \in B(\Ha): \Phi^*(X) = X\}$, note that we have $\Phi^*(\A_i) = \A_i$ as a consequence of \eqref{eq:broad-A} and \eqref{eq:broad-B}, so $\testa \subset \fix(\Phi^*)^+$. We will now treat $\fix(\Phi^*)$ as a vector space and $\fix(\Phi^*)^+$ as a proper cone in $\fix(\Phi^*)$. We then have the dual vector space $\fix(\Phi^*)^*$ which contains the dual cone $\fix(\Phi^*)^{*+}$. We can now treat $\Pi: B(\Ha) \to \fix(\Phi^*)$ as a positive linear map and define $\iota: \fix(\Phi^*) \to B(\Ha)$ as the trivial embedding of $\fix(\Phi^*)$ into $B(\Ha)$, so that we have that $\Pi \circ\iota = \mathrm{id}_{\fix(\Phi^*)}$ is the identity map on $\fix(\Phi^*)$. The adjoint positive linear maps are $\Pi^*: \fix(\Phi^*)^* \to B(\Ha)$ and $\iota^*: B(\Ha) \to \fix(\Phi^*)^*$ and we again have $\iota^* \circ \Pi^* = \mathrm{id}_{\fix(\Phi^*)^*}$. Here $\fix(\Phi^*)^{*+}$ can be understood as the cone generated by states of some suitable operational theory \cite{muller2021probabilistic,Plavala21} and $\fix(\Phi^*)^+$ as the dual cone generated by the effect algebra.

Let us now construct the map $\Lambda_{\fix(\Phi^*)^*}: \fix(\Phi^*) \to \fix(\Phi^*)^{\otimes 2}$ as
\begin{equation}
\Lambda_{\fix(\Phi^*)^*} = (\iota^* \otimes \iota^*) \circ \Lambda \circ \Pi^*
\end{equation}
and we will show that $\Lambda_{\fix(\Phi^*)^*}$ is a perfect broadcasting map on $\fix(\Phi^*)^{*+}$, which then the no-broadcasting theorem \cite{BaBaLeWi07} implies that $\fix(\Phi^*)^{*+}$ is a classical cone, i.e., that any base of $\fix(\Phi^*)^{*+}$ must be a simplex. So let $x \in \fix(\Phi^*)^+$ and $y \in \fix(\Phi^*)^{*+}$, then we have
\begin{equation}
\langle \Lambda_{\fix(\Phi^*)^*}(y), x \otimes \id \rangle = \tr{\Pi^*(y) \Lambda^*(\iota(x) \otimes \id)} = \tr{\Pi^*(y) \iota(x)} = \langle y, x \rangle
\end{equation}
where $\langle \cdot, \cdot \rangle$ denotes the pairing between $\fix(\Phi^*)$ and $\fix(\Phi^*)^*$; $\langle \id \otimes x, \Lambda_{\fix(\Phi^*)^*}(y) \rangle = \langle x, y \rangle$ follows in the same way. Since $\fix(\Phi^*)^{*+}$ is generated by a simplex, it follows that there is basis of $g_\lambda \in \fix(\Phi^*)^+$ such that every $x \in \fix(\Phi^*)^+$ is a positive linear combination of $g_\lambda$, i.e., $x = \sum_\lambda \alpha_\lambda g_\lambda$, $\alpha_\lambda \geq 0$, $\id = \sum_\lambda g_\lambda$, and there are $s_\lambda \in \fix(\Phi^*)^{*+}$ such that $\tr{\Pi^*(s_\lambda)} = 1$ and $\langle s_\mu, g_\lambda \rangle = \delta_{\lambda \mu}$. Then let $\G_\lambda = \iota(g_\lambda)$, then we have that $\no{\G_\lambda} = 1$ since $\langle s_\lambda, g_\lambda \rangle = \tr{\Pi^*(s_\lambda) \G_\lambda} = 1$ where $\Pi^*(s_\lambda) \in \state$. Moreover for any $\A \in \testa$, since $\A_i \in \fix(\Phi^*)$, we have $\A_i = \sum_\lambda p(i|\A, \lambda) \G_\lambda$, $\sum_i p(i|\A, \lambda) = 1$ follows from the linear independence of $\G_\lambda$.
\end{proof}

\section{Proof of Theorem~\ref*{thm:tao}} \label{appendix:proofTao}
Let $\tests \subseteq \state$ be a set of states and $\testa, \testb \subseteq \obs$ be two sets of POVMs. The triple $(\tests,\testa,\testb)$ is broadcastable if there exists a channel $\Lambda: \state \to \state \otimes \state$ such that \eqref{eq:broad-A} and \eqref{eq:broad-B} hold for all $\A \in \testa$, $\B \in \testb$, $\rho \in \tests$. In other words, there exist two compatible channels $\Phi_1=\ptr{2}{\Lambda(\rho)}$, $\Phi_2=\ptr{1}{\Lambda(\rho)}$ such that
\begin{align}
\A(\Phi_1(\rho)) & = \A(\rho) \\
\B(\Phi_2(\rho)) & = \B(\rho)
\end{align}
for all $\A \in \testa$, $\B \in \testb$, $\rho \in \tests$. Assume that $\testb$ is information complete, so that the second equality becomes
\begin{equation}
\Phi_2(\rho)=\rho.
\end{equation}

Since $\tests$ is a convex set, there is some element $\sigma\in \tests$ such that its support contains the supports of all other states in $\tests$. Let $P = \mathrm{supp}(\sigma)$ be the projection onto the support of $\sigma$. Note that since $\sigma\in \fix(\Phi_2)$, we have $\Phi_2(PB(\Ha)P)\subseteq PB(\Ha)P$. Let us replace the channels $\Phi_1$ and $\Phi_2$ with their restrictions to $PB(\Ha)P\equiv B(P\Ha)$, then $\Phi_2$ is a channel $B(P\Ha)\to B(P\Ha)$ compatible with the channel $\Phi_1: B(P\Ha) \to B(\Ha)$ and $\tests \subseteq \fix(\Phi_2)$, so that $\fix(\Phi_2)$ contains a full rank state. According to Lemma \ref{lemma:saa-fixedPoint}, there is a unital completely positive idempotent map whose range is the set of fixed points $\fix(\Phi_2^*)$. Let $\Pi_2$ denote the adjoint of this map, then it is easily seen that $\Pi_2$ is an  idempotent channel such that  $\fix(\Pi_2)=\fix(\Phi_2)$. Since compatibility is preserved by post-processings, $\Pi_2=\Pi_2\circ\Phi_2$ is compatible with $\Phi_1$. We may therefore assume that $\Phi_2$ is an idempotent channel with a full rank fixed state.

Let us recall some well known facts, for the convenience of the reader. Let $\Pi:B(\Ha)\to B(\Ha)$ be an idempotent channel which has a full rank fixed state $\sigma$. First, note that the adjoint map $\Pi^*:B(\Ha)\to B(\Ha)$ is an idempotent unital completely-positive map whose range is a subalgebra of $B(\Ha)$. Indeed, the range of $\Pi^*$ is clearly a self-adjoint linear subspace containing $\id$, it is therefore enough to prove that $\Pi^*(X^*X)=X^*X$ for all $X =\Pi^*(X)$. By the Schwarz inequality for unital completely-positive maps \cite[Prop. 3.3]{paulsen2002completely}, we have
\begin{equation}
\Pi^*(X^*X) \geq \Pi^*(X)^*\Pi^*(X) = X^*X
\end{equation}
and since $\sigma \in \fix(\Pi)$, we have $\tr{\sigma(\Pi^*(X^*X) - X^*X)} = 0$, so that $\Pi^*(X^*X)=X^*X$. Moreover $\Pi^*$ is a conditional expectation, which means that it has the conditional expectation property
\begin{equation}
\Pi^*(\Pi^*(X)Z\Pi^*(Y))=\Pi^*(X)\Pi^*(Z)\Pi^*(Y)
\end{equation}
for $X,Y, Z\in B(\Ha)$. This can be seen from the fact that $\Pi^*$ is multiplicative on its range and \cite[Thm. 3.18]{paulsen2002completely}.

The next result extends the well known equivalence of commutativity and compatibility for projective measurements.

\begin{proposition}\label{prop:condexp_compatible} Let $\Pi:B(\Ha)\to B(\Ha)$ be an idempotent channel whose adjoint $\Pi^*$ is a conditional expectation and let $\Phi: B(\Ha)\to B(\Ka)$ be any channel. Then $\Phi$ is compatible with $\Pi$ if and only if $\Phi^*$ and $\Pi^*$ have commuting ranges.

\end{proposition}

\begin{proof} Let $\Lambda:B(\Ha)\to B(\Ka\otimes \Ha)$ be the joint channel for $\Phi$
and $\Pi$, so that the adjoint map satisfies
\begin{equation}
\Lambda^*(X\otimes I)=\Phi^*(X),\quad \Lambda^*(I\otimes Y)=\Pi^*(Y),\qquad X\in B(\Ka),\
Y\in B(\Ha).
\end{equation}
Since the range of $\Pi^*$ is a subalgebra in $B(\Ha)$, it is generated by projections it contains. It is therefore enough to show that any projection $Q=\Pi^*(Q)$ commutes with all elements of the form $\Phi^*(X)$, $X\in B(\Ka)$. We may also restrict to effects, $0\le X\le I_\Ka$, since these generate $B(\Ka)$. In this case $\Phi^*(X)$ is an effect as well and using the joint channel, it is easily seen that it is compatible with $Q$. Since $Q$ is a projection, this shows that $Q$ must commute with $\Phi^*(X)$.

For the converse, assume that $\Phi^*$ and $\Pi^*$ have commuting ranges, then the map
\begin{equation}
\Lambda^*: B(\Ka\otimes \Ha)\to B(\Ha),\quad X\otimes Y\mapsto \Phi^*(X)\Pi^*(Y),\qquad
X\in B(\Ka),\ Y\in B(\Ha)
\end{equation}
is unital and completely positive (e.g. \cite[Prop. VI.4.23]{TOA1}), and its adjoint is a joint channel for $\Phi$ and $\Pi$.
\end{proof}

\thmTao*
\begin{proof}
We see by Proposition \ref{prop:condexp_compatible} and the paragraph above it that we may assume that $\Phi_1$ and $\Phi_2$ are compatible channels such that the adjoint of $\Phi_2^*$ is a conditional expectation, so that the ranges of the adjoints $\Phi_1^*$ and $\Phi_2^*$ must commute.

Let $A\in B(\Ha)$ and $\rho\in \tests$. Then
\begin{equation}
\tr{A \Phi_1(\rho)} = \tr{\Phi_1^*(A) \rho} = \tr{\Phi_1^*(A) \Phi_2(\rho)} =\tr{\Phi_2^*(\Phi_1^*(A))\rho} \, .
\end{equation}
By the conditional expectation property of $\Phi^*_2$ we have for any $C=\Phi_2^*(C)$,
\begin{equation}
\Phi_2^*(\Phi_1^*(A))C = \Phi_2^*(\Phi_1^*(A)C) = \Phi_2^*(C\Phi_1^*(A)) =
C\Phi_2^*(\Phi_1^*(A)),
\end{equation}
so that $Z=\Phi_2^*(\Phi_1^*(A))$ commutes with all other elements in the range of
$\Phi_2^*$ . For any $X \in B(P\Ha)$ and $\rho \in \tests$:
\begin{equation}
\tr{XZ\rho} = \tr{XZ \Phi_2(\rho)} = \tr{\Phi_2^*(XZ) \rho} = \tr{\Phi_2^*(X)Z \rho} =
\tr{Z\Phi_2^*(X)\rho} = \tr{\Phi_2^*(ZX) \rho} = \tr{X \rho Z} \, ,
\end{equation}
so that $Z$ commutes with $\rho$. For $\A \in \testa$, it is now enough to put $\tilde{\A}
= \Phi_2^*(\Phi_1^*(\A))$, then $\tilde{\A}$ commutes with all $\rho \in \tests$ and we have
\begin{equation}
\tilde{\A}(\rho) = \Phi_2^*(\Phi_1^*(\A))(\rho) = \Phi_1^*(\A)(\rho) = \A(\rho).
\end{equation}
\end{proof}

\end{document}